\documentclass[notitlepage,11pt]{article}%
\usepackage{latexsym,epsfig,amsmath}
\usepackage{epsfig,graphicx}
\usepackage{amsfonts}
\usepackage{tabularx}
\usepackage{rotating}
\usepackage{hyperref}
\usepackage{amsmath}
\usepackage{amssymb}
\usepackage{graphicx}
\usepackage{caption}
\usepackage{footmisc}
\usepackage{color}
\usepackage[sort]{natbib}%
\providecommand{\U}[1]{\protect \rule{.1in}{.1in}}
\newtheorem{lemma}{Lemma}

\newenvironment{proof}[1][Proof]{\textbf{#1.} }{\  \rule{0.5em}{0.5em}}

\hypersetup{colorlinks=true,citecolor=blue}
\begin{document}

\title{Some Finite-Sample Results on the Hausman Test}
\author{Jinyong Hahn\thanks{Department of Economics, UCLA, Los Angeles, CA 90095-1477
USA. Email: hahn@econ.ucla.edu.}\\UCLA
\and Zhipeng Liao\thanks{Department of Economics, UCLA, Los Angeles, CA 90095-1477
USA. Email: zhipeng.liao@econ.ucla.edu.}\\UCLA
\and Nan Liu\thanks{Wang Yanan Institute for Studies in Economics and School of
Economics, Xiamen University, Xiamen, Fujian, 36100. Email:
nanliu@xmu.edu.cn.}\\Xiamen University
\and Shuyang Sheng\thanks{Department of Economics, UCLA, Los Angeles, CA 90095-1477
USA. Email: ssheng@econ.ucla.edu}\\UCLA}
\date{\today }
\maketitle

\begin{abstract}
This paper shows that the endogeneity test using the control function approach
in linear instrumental variable models is a variant of the Hausman test.
Moreover, we find that the test statistics used in these tests can be
numerically ordered, indicating their relative power properties in finite samples.

\bigskip


\end{abstract}

\section{Introduction}

This paper investigates the endogeneity test for potentially endogenous
regressors in linear instrumental variable (IV) models that utilizes the
control function (CF) approach. Specifically, we focus on the following model:%
\begin{align}
y_{2} &  =y_{1}^{\top}\beta+z_{1}^{\top}\gamma+\varepsilon=x^{\top}%
\theta+\varepsilon,\label{s_m}\\
y_{1} &  =z_{1}^{\top}\pi_{1}+z_{2}^{\top}\pi_{2}+v=z^{\top}\pi+v,\label{r_m}%
\end{align}
where $x\equiv(y_{1}^{\top},z_{1}^{\top})^{\top}$, $y_{1}$ represents
potentially endogenous regressors, $\varepsilon$ and $v$ are unobserved error
terms in the structural and reduced-form equations respectively, and
$z\equiv(z_{1}^{\top},z_{2}^{\top})^{\top}$ represents exogenous variables
which satisfy:%
\begin{equation}
E[z\varepsilon]=0\text{ \  \  \  \ and \  \  \  \ }E[zv^{\top}]=0.\label{moment_c}%
\end{equation}
Since $z_{2}$ is excluded from the structure equation (\ref{s_m}) and
satisfies the orthogonality conditions in (\ref{moment_c}), it serves as
instruments for $y_{1}$.

One widely used method for testing the endogeneity of $y_{1}$ is the Hausman
test
%TCIMACRO{\TeXButton{hausman1978}{\citep{hausman1978}}}%
%BeginExpansion
\citep{hausman1978}%
%EndExpansion
. This test compares the ordinary least squares (OLS) estimator $\hat{\beta
}_{ols}$ of $\beta$ against the two-stage least squares (2SLS) estimator
$\hat{\beta}_{2sls}$. The test rejects the null hypothesis that $y_{1}$ is
exogenous if the difference between $\hat{\beta}_{ols}$ and $\hat{\beta
}_{2sls}$ exceeds a certain threshold determined by the significance level of
the test.

As an alternative to the Hausman test in the IV approach, we can also apply
the CF approach to obtain a consistent estimator of $\beta$ and test for the
endogeneity of $y_{1}$. Specifically, we run an OLS regression of the
following model%
\begin{equation}
y_{2}=x^{\top}\theta+v^{\top}\rho+u, \label{m_cv}%
\end{equation}
where we replace $v$ by the estimated residual $\hat{v}$ in the OLS regression
of the reduced-form equation (\ref{r_m}), and obtain estimators $\hat{\theta
}_{cf}$ and $\hat{\rho}_{cf}$. Then we test the endogeneity of $y_{1}$ using
the Wald test for the null hypothesis $H_{0}:\rho=0$.

This paper makes several contributions to the existing literature. First, we
demonstrate that $\hat{\rho}_{cf}$ is a linear transformation of $\hat{\beta
}_{ols}-\hat{\beta}_{2sls}$, thereby elucidating the connection between the
Wald test for $H_{0}:\rho=0$ in the CF approach and the Hausman test in the IV
approach. Second, we show that the Wald test differs from the Hausman test
primarily in how the asymptotic variances of $\hat{\beta}_{2sls}$ and
$\hat{\beta}_{ols}$ are estimated, thus representing a variant of the Hausman
test. Third, our analysis reveals that the Wald test statistic is numerically
larger than the Hausman test statistics, which rely on $\hat{\beta}_{ols}$ or
$\hat{\beta}_{2sls}$ to obtain estimators of\ the asymptotic variances of
$\hat{\beta}_{ols}\ $and $\hat{\beta}_{2sls}$. Since the Wald test employs the
same critical value as the Hausman test, it exhibits larger statistical power
in finite samples. These findings are established without imposing restrictive
assumptions on either the null or alternative hypotheses.

The remainder of the paper is structured as follows. Section \ref{Sec2}
introduces the test statistics employed in both the Hausman test and the Wald
test. Section \ref{Sec3} establishes a numerical order among the test
statistics introduced in Section \ref{Sec2} and discusses its implications for
relative power properties of the tests in finite samples. Section \ref{Sec4}
concludes the paper. The proofs are presented in the Appendix.

\textit{Notation.} We use $a\equiv b$ to indicate that $a$ is defined as $b$.
For any real vector $a$, $d_{a}$ denotes the dimension of $a$. For any
positive integer $k$, $I_{k}$ denotes the $k\times k$ identity matrix. For any
$k_{1}\times k_{2}$ matrix $A$, $A^{\top}$ denotes the transpose of $A$,
$P_{A}\equiv A(A^{\top}A)^{-1}A^{\top}$ and $M_{A}\equiv I_{k_{1}}-A(A^{\top
}A)^{-1}A^{\top}$ as long as $A^{\top}A$ is non-singular. For any square
matrix $A$, $A>0$ means $A$ is positive definite.

\section{Testing for Endogeneity\label{Sec2}}

We begin by formulating the Hausman test for assessing the endogeneity of
$y_{1}$ in the model specified by (\ref{s_m}) and (\ref{r_m}). Suppose we have
a random sample $\left \{  y_{1,i},y_{2,i},z_{i}\right \}  _{i=1}^{n}$, where
$z_{i}\equiv(z_{1,i}^{\top},z_{2,i}^{\top})^{\top}$ for $i=1,\ldots,n$. Recall
that $x_{i}\equiv(y_{1,i}^{\top},z_{1,i}^{\top})^{\top}$, $i=1,\ldots,n$,
denotes the regressors in (\ref{s_m}). Let $X\equiv(x_{1},\ldots,x_{n})^{\top
}$, $Y_{1}\equiv(y_{1,1},\ldots,y_{1,n})^{\top}$, $Y_{2}\equiv(y_{2,1}%
,\ldots,y_{2,n})^{\top}$, and\ $Z\equiv(z_{1},\ldots,z_{n})^{\top}$. The OLS
and 2SLS estimators of $\theta$ in (\ref{s_m}) are given by%
\[
\hat{\theta}_{ols}\equiv(X^{\top}X)^{-1}X^{\top}Y_{2}\  \  \  \  \  \  \  \text{and}%
\  \  \  \  \  \  \  \hat{\theta}_{2sls}\equiv(X^{\top}P_{Z}X)^{-1}X^{\top}P_{Z}Y_{2}%
\]
respectively. Let $\hat{\beta}_{ols}$ and $\hat{\beta}_{2sls}$ denote the
leading $d_{y_{1}}\times1$ subvectors of $\hat{\theta}_{ols}$ and $\hat
{\theta}_{2sls}$ respectively. The Hausman test statistic can be characterized
by%
\begin{equation}
t_{H,n}(\hat{\sigma}_{1}^{2},\hat{\sigma}_{2}^{2})\equiv(\hat{\beta}%
_{ols}-\hat{\beta}_{2sls})^{\top}(\hat{\sigma}_{1}^{2}(\hat{Y}_{1}^{\top
}M_{Z_{1}}\hat{Y}_{1})^{-1}-\hat{\sigma}_{2}^{2}(Y_{1}^{\top}M_{Z_{1}}%
Y_{1})^{-1})^{-1}(\hat{\beta}_{ols}-\hat{\beta}_{2sls}), \label{H_test}%
\end{equation}
where $\hat{Y}_{1}\equiv P_{Z}Y_{1}$, $\hat{\sigma}_{1}^{2}(\hat{Y}_{1}^{\top
}M_{Z_{1}}\hat{Y}_{1})^{-1}$ and $\hat{\sigma}_{2}^{2}(Y_{1}^{\top}M_{Z_{1}%
}Y_{1})^{-1}$ are estimators of the asymptotic variances of $\hat{\beta
}_{2sls}\ $and $\hat{\beta}_{ols}$ respectively, and $\hat{\sigma}_{1}^{2}$
and $\hat{\sigma}_{2}^{2}$ are (possibly different) estimates of the variance
$\sigma_{\varepsilon}^{2}$ of the error term $\varepsilon$ in the structural
equation (\ref{s_m}).

In practice, $\sigma_{\varepsilon}^{2}$ can be estimated by the sample
variance of the fitted residual in the OLS estimation:%
\[
\hat{\sigma}_{ols}^{2}\equiv n^{-1}(Y_{2}-X\hat{\theta}_{ols})^{\top}%
(Y_{2}-X\hat{\theta}_{ols}),
\]
or through its counterpart in the 2SLS estimation:
\[
\hat{\sigma}_{2sls}^{2}\equiv n^{-1}(Y_{2}-X\hat{\theta}_{2sls})^{\top}%
(Y_{2}-X\hat{\theta}_{2sls}),
\]
resulting in three popular versions of the Hausman test with test statistics
$t_{H_{1}}\equiv t_{H,n}(\hat{\sigma}_{ols}^{2},\hat{\sigma}_{ols}^{2})$,
$t_{H_{2}}\equiv t_{H,n}(\hat{\sigma}_{2sls}^{2},\hat{\sigma}_{2sls}^{2})$,
and $t_{H_{3}}\equiv t_{H,n}(\hat{\sigma}_{2sls}^{2},\hat{\sigma}_{ols}^{2})$
respectively (%
%TCIMACRO{\TeXButton{wooldridge2010}{\citealp[Section 6.3.1]{wooldridge2010}}}%
%BeginExpansion
\citealp[Section 6.3.1]{wooldridge2010}%
%EndExpansion
;
%TCIMACRO{\TeXButton{Baum2003}{\citealp{Baum2003}}}%
%BeginExpansion
\citealp{Baum2003}%
%EndExpansion
).

Alternatively, we can apply the CF approach and obtain the estimators of
$\theta$ and $\rho$ in (\ref{m_cv}) as
\begin{equation}
(\hat{\theta}_{cf}^{\top},\hat{\rho}_{cf}^{\top})^{\top}\equiv((X,\hat
{V})^{\top}(X,\hat{V}))^{-1}(X,\hat{V})^{\top}Y_{2}, \label{CV_est}%
\end{equation}
where $\hat{V}\equiv M_{Z}Y_{1}$, and then test the endogeneity of \ $y_{1}$
using the Wald test with the test statistic%
\begin{equation}
t_{CF}\equiv \hat{\rho}_{cf}^{\top}\left(  \widehat{Asv}(\hat{\rho}%
_{cf})\right)  ^{-1}\hat{\rho}_{cf}, \label{W_test}%
\end{equation}
where $\widehat{Asv}(\hat{\rho}_{cf})$ denotes an estimator of the asymptotic
variance of $\hat{\rho}_{cf}$. Applying the partitioned regression formula
(which is also known as the Frisch-Waugh-Lovell Theorem, see, e.g.,
%TCIMACRO{\TeXButton{davidson1993}{\citet[Section 1.4]{davidson1993}}}%
%BeginExpansion
\citet[Section 1.4]{davidson1993}%
%EndExpansion
) to (\ref{CV_est}) yields%
\begin{equation}
\hat{\rho}_{cf}=(\hat{V}^{\top}M_{X}\hat{V})^{-1}(\hat{V}^{\top}M_{X}%
Y_{2})=\rho+(\hat{V}^{\top}M_{X}\hat{V})^{-1}[\hat{V}^{\top}M_{X}U-\hat
{V}^{\top}M_{X}(\hat{V}-V)\rho], \label{CV_est_exp}%
\end{equation}
where\ $U\equiv(u_{1},\ldots,u_{n})^{\top}$, $V\equiv(v_{1},\ldots
,v_{n})^{\top}$, and the second equality follows by (\ref{m_cv}). Although the
regressor $v$ is estimated and hence the estimation error $\hat{V}-V$ should
be taken into account when calculating $\widehat{Asv}(\hat{\rho}_{cf})$, the
expansion in\ (\ref{CV_est_exp}) shows that this estimation error can be
ignored under the null\ that $\rho=0$. Therefore, in the rest of the paper we
use the estimator%
\begin{equation}
\widehat{Asv}(\hat{\rho}_{cf})=\hat{\sigma}_{u}^{2}(\hat{V}^{\top}M_{X}\hat
{V})^{-1}, \label{CV_sigma}%
\end{equation}
where $\hat{\sigma}_{u}^{2}$ is the sample variance of the fitted residual in
the OLS regression of (\ref{m_cv})%
\begin{equation}
\hat{\sigma}_{u}^{2}\equiv n^{-1}(Y_{2}-X\hat{\theta}_{cf}-\hat{V}\hat{\rho
}_{cf})^{\top}(Y_{2}-X\hat{\theta}_{cf}-\hat{V}\hat{\rho}_{cf}).
\label{CV_sigma_u}%
\end{equation}
As we shall see in the next section, choosing this estimator makes the Wald
statistic $t_{CF}$ comparable to\ $t_{H,n}(\hat{\sigma}_{1}^{2},\hat{\sigma
}_{2}^{2})$.

Throughout this paper, we assume that $X^{\top}X$, $X^{\top}P_{Z}X$, and
$(X,\hat{V})^{\top}(X,\hat{V})$ are non-singular so that the estimators
$\hat{\theta}_{ols}$, $\hat{\theta}_{2sls}$, $\hat{\theta}_{cf}$ and
$\hat{\rho}_{cf}$ are well defined. Under these primitive conditions, we have
$Y_{1}^{\top}M_{Z_{1}}Y_{1}>0$, $\hat{Y}_{1}^{\top}M_{Z_{1}}\hat{Y}_{1}>0$ and%
\[
Y_{1}^{\top}M_{Z_{1}}Y_{1}-\hat{Y}_{1}^{\top}M_{Z_{1}}\hat{Y}_{1}=Y_{1}^{\top
}M_{Z}Y_{1}>0,
\]
which further implies that%
\begin{equation}
(\hat{Y}_{1}^{\top}M_{Z_{1}}\hat{Y}_{1})^{-1}-(Y_{1}^{\top}M_{Z_{1}}%
Y_{1})^{-1}>0. \label{PD_1}%
\end{equation}
In view of (\ref{PD_1}) and the definitions of $t_{H_{j}}$ ($j=1,2,3$) and
$t_{CF}$, we also assume that $\hat{\sigma}_{ols}^{2}$, $\hat{\sigma}%
_{2sls}^{2}$ and $\hat{\sigma}_{u}^{2}$ are strictly positive so that these
test statistics are well defined.

Under the null that $y_{1}$ is exogenous, along with other regularity
conditions, one can establish that the asymptotic distributions of $t_{H_{j}}$
($j=1,2,3$) and $t_{CF}$ are Chi-square with $k_{1}$ degrees of freedom
(denoted as $\chi^{2}(k_{1})$). Consequently, the Hausman tests and the Wald
test reject the null hypothesis if the corresponding test statistic exceeds
the $1-\alpha$ quantile of $\chi^{2}(k_{1})$, where $\alpha$ denotes the
significance level. As we shall see in the next section, the test statistics
$t_{H_{j}}$ ($j=1,2,3$) and $t_{CF}$ can be numerically ordered, indicating
their relative rejection properties in finite samples.

\section{Main Results\label{Sec3}}

Our first objective is to establish a numerical relationship between the OLS
and 2SLS estimators and the estimators in the CF approach. This result serves
as a foundation for further investigating the numerical order among $t_{H_{j}%
}$ ($j=1,2,3$) and $t_{CF}$ in finite samples.

\begin{lemma}
\label{lm1} The estimators in the CF approach satisfy%
\begin{equation}
\left(
\begin{array}
[c]{c}%
\hat{\theta}_{cf}\\
\hat{\rho}_{cf}%
\end{array}
\right)  =\left(
\begin{array}
[c]{c}%
\hat{\theta}_{2sls}\\
(Y_{1}^{\top}M_{Z}Y_{1})^{-1}(Y_{1}^{\top}M_{Z_{1}}Y_{1})(\hat{\beta}%
_{ols}-\hat{\beta}_{2sls})
\end{array}
\right)  . \label{lm1_1}%
\end{equation}
Moreover, the Wald test statistic\ $t_{CF}$ satisfies
\begin{equation}
t_{CF}=(\hat{\beta}_{ols}-\hat{\beta}_{2sls})^{\top}\left(  \hat{\sigma}%
_{u}^{2}(\hat{Y}_{1}^{\top}M_{Z_{1}}\hat{Y}_{1})^{-1}-\hat{\sigma}_{u}%
^{2}(Y_{1}^{\top}M_{Z}Y_{1})^{-1}\right)  ^{-1}(\hat{\beta}_{ols}-\hat{\beta
}_{2sls}), \label{lm1_2}%
\end{equation}
where $\hat{\sigma}_{u}^{2}$ is defined in (\ref{CV_sigma_u}).
\end{lemma}

The lemma above carries two important implications. First, $\hat{\theta}_{cf}$
is numerically equivalent to $\hat{\theta}_{2sls}$, implying that $\hat
{\theta}_{cf}$ shares the same standard error as $\hat{\theta}_{2sls}$. This
finding has been well recognized in the literature since at least
\cite{hausman1978} (see also
%TCIMACRO{\TeXButton{davidson1993}{\citet[Section 7.9]{davidson1993}} }%
%BeginExpansion
\citet[Section 7.9]{davidson1993}
%EndExpansion
and
%TCIMACRO{\TeXButton{wooldridge2010}{\citet[Problem 5.1]{wooldridge2010}}}%
%BeginExpansion
\citet[Problem 5.1]{wooldridge2010}%
%EndExpansion
). Second, $\hat{\rho}_{cf}$ is a linear transformation of $\hat{\beta}%
_{ols}-\hat{\beta}_{2sls}$, establishing a connection between the Wald test
based on $t_{CF}$ and the Hausman tests based on\ $t_{H,n}(\hat{\sigma}%
_{1}^{2},\hat{\sigma}_{2}^{2})$. To the best of our knowledge, this numerical
relationship is a novel contribution to the literature. The expression of
$t_{CF}$ in (\ref{lm1_2}) further suggests that $t_{CF}$ is a special case of
$t_{H,n}(\hat{\sigma}_{1}^{2},\hat{\sigma}_{2}^{2})$ with $\hat{\sigma}%
_{1}^{2}=\hat{\sigma}_{2}^{2}=\hat{\sigma}_{u}^{2}$.

Since $t_{H_{j}}$ ($j=1,2,3$) and $t_{CF}$ differ only in how the variance
estimators $\hat{\sigma}_{1}^{2}\ $and $\hat{\sigma}_{2}^{2}$ in
(\ref{H_test}) are calculated, their relative performances are determined by
the differences of these variance estimators. Intuitively,\ $\hat{\sigma}%
_{u}^{2}$ should not be larger than $\hat{\sigma}_{ols}^{2}$ since, compared
with the OLS estimation of (\ref{s_m}), the CF approach includes an extra
regressor $v$ in (\ref{m_cv}), and the resulting $R^{2}$ should not be
smaller. Moreover, we have $\hat{\sigma}_{ols}^{2}\leq \hat{\sigma}_{2sls}^{2}$
by the definitions of OLS and 2SLS estimation. Therefore, we expect a weak
order among these variance estimators: $\hat{\sigma}_{u}^{2}\leq \hat{\sigma
}_{ols}^{2}\leq \hat{\sigma}_{2sls}^{2}$, which together with the definitions
of $t_{H_{j}}$ ($j=1,2,3$) and $t_{CF}$ implies a weak numerical order among
the test statistics:\ $t_{CF}\geq t_{H_{1}}\geq t_{H_{2}}\geq t_{H_{3}}$. Our
next lemma establishes the exact relationships among the variance estimators
$\hat{\sigma}_{u}^{2}$, $\hat{\sigma}_{ols}^{2}$ and $\hat{\sigma}_{2sls}^{2}%
$, enabling us to obtain a strong numerical order among them as well as among
the test statistics $t_{H_{j}}$ ($j=1,2,3$) and $t_{CF}$.

\begin{lemma}
\label{lm2}\ Let\ $H_{n}\equiv(n\hat{\sigma}_{2sls}^{2})^{-1}(\hat{\beta
}_{ols}-\hat{\beta}_{2sls})^{\top}(Y_{1}^{\top}M_{Z_{1}}Y_{1})(\hat{\beta
}_{ols}-\hat{\beta}_{2sls})$. Then%
\begin{equation}
\hat{\sigma}_{u}^{2}=\hat{\sigma}_{ols}^{2}\left(  1-\frac{t_{H_{1}}}%
{n}\right)  =\hat{\sigma}_{2sls}^{2}\left(  1-\frac{t_{H_{2}}}{n}%
-H_{n}\right)  . \label{lm2_1}%
\end{equation}

\end{lemma}

Since $t_{H_{j}}$ ($j=1,2,3$) and $H_{n}$ are non-negative, from (\ref{lm2_1})
we immediately obtain $\hat{\sigma}_{2sls}^{2}\geq \hat{\sigma}_{ols}^{2}%
\geq \hat{\sigma}_{u}^{2}$, which implies that $t_{CF}\geq t_{H_{1}}\geq
t_{H_{2}}\geq t_{H_{3}}$. Moreover, when $\hat{\beta}_{ols}-\hat{\beta}%
_{2sls}\neq0$, $t_{H_{j}}$ and $H_{n}$ are strictly positive. In this case, we
can deduce from (\ref{lm2_1}) that $\hat{\sigma}_{2sls}^{2}>\hat{\sigma}%
_{ols}^{2}>\hat{\sigma}_{u}^{2}$, and a strong numerical order among
$t_{H_{j}}$ ($j=1,2,3$) and $t_{CF}$, which is summarized in the lemma below.

\begin{lemma}
\label{lm3} Suppose that $\hat{\beta}_{ols}-\hat{\beta}_{2sls}\neq0$. Then we
have $t_{CF}>t_{H_{1}}>t_{H_{2}}>t_{H_{3}}$.
\end{lemma}

Lemma \ref{lm3} demonstrates that in finite samples, the endogeneity test
based on $t_{CF}$ has the largest rejection probability compared with
$t_{H_{j}}$\ ($j=1,2,3$), although these four tests are asymptotically
equivalent under the null hypothesis and local alternatives where $\hat{\beta
}_{ols}-\hat{\beta}_{2sls}=o_{p}(1)$.

\section{Conclusion\label{Sec4}}

This paper explores the endogeneity test using the CF approach in linear IV
models. We find that the OLS estimator of the coefficients of the generated CF
is a linear transformation of the difference between the OLS and 2SLS
estimators of the coefficients of endogenous regressors. This finding allows
us to demonstrate that the endogeneity test using the CF approach is a variant
of the Hausman test. In addition, we establish a numerical order among the
test statistics employed in the Hausman tests and the endogeneity test using
the CF approach. This numerical order highlights that the endogeneity test
using the CF approach exhibits the highest rejection probability in finite samples.

\bigskip

{\small
\bibliographystyle{econometrica}
\bibliography{Hausman}
}

\newpage

\appendix

\begin{center}
{\huge Appendix}
\end{center}

\bigskip

\section{Proofs}

\begin{proof}
[Proof of Lemma \ref{lm1}]By partitioned regression/partialling out formula,
\begin{equation}
\left(
\begin{array}
[c]{c}%
\hat{\theta}_{cf}\\
\hat{\rho}_{cf}%
\end{array}
\right)  =\left(
\begin{array}
[c]{c}%
\left(  X^{\top}M_{\hat{V}}X\right)  ^{-1}X^{\top}M_{\hat{V}}Y_{2}\\
(\hat{V}^{\top}M_{X}\hat{V})^{-1}\hat{V}^{\top}M_{X}Y_{2}%
\end{array}
\right)  . \label{P_L1_0}%
\end{equation}
Since $X^{\top}\hat{V}=X^{\top}M_{Z}Y_{1}=(Y_{1}^{\top}M_{Z}Y_{1},0_{d_{y_{1}%
}\times d_{z_{1}}})^{\top}$\ and\ $\hat{V}^{\top}\hat{V}=Y_{1}^{\top}%
M_{Z}Y_{1}$, we have%
\[
X^{\top}M_{\hat{V}}X=X^{\top}X-\left(
\begin{array}
[c]{cc}%
Y_{1}^{\top}M_{Z}Y_{1} & 0_{d_{y_{1}}\times d_{z_{1}}}\\
0_{d_{z_{1}}\times d_{y_{1}}} & 0_{d_{z_{1}}\times d_{z_{1}}}%
\end{array}
\right)  =\left(
\begin{array}
[c]{cc}%
Y_{1}^{\top}P_{Z}Y_{1} & Y_{1}^{\top}Z_{1}\\
Z_{1}^{\top}Y_{1} & Z_{1}^{\top}Z_{1}%
\end{array}
\right)  =X^{\top}P_{Z}X
\]
and%
\[
X^{\top}M_{\hat{V}}Y_{2}=X^{\top}Y_{2}-\left(
\begin{array}
[c]{c}%
Y_{1}^{\top}M_{Z}Y_{2}\\
0_{d_{z_{1}}\times1}%
\end{array}
\right)  =\left(
\begin{array}
[c]{c}%
Y_{1}^{\top}P_{Z}Y_{2}\\
Z_{1}^{\top}Y_{2}%
\end{array}
\right)  =X^{\top}P_{Z}Y_{2},
\]
which together with (\ref{P_L1_0}) implies that%
\begin{equation}
\hat{\theta}_{cf}=(X^{\top}P_{Z}X)^{-1}X^{\top}P_{Z}Y_{2}=\hat{\theta}_{2sls}.
\label{P_L1_1}%
\end{equation}
By the definitions of $\hat{\theta}_{ols}$ and $\hat{\theta}_{2sls}$, it is
easy to show that%
\begin{equation}
\hat{\beta}_{ols}=(Y_{1}^{\top}M_{Z_{1}}Y_{1})^{-1}Y_{1}^{\top}M_{Z_{1}}%
Y_{2}\text{ \  \  \ and \  \  \ }\hat{\beta}_{2sls}=(Y_{1}^{\top}(P_{Z}-P_{Z_{1}%
})Y_{1})^{-1}Y_{1}^{\top}(P_{Z}-P_{Z_{1}})Y_{2}. \label{P_L1_2}%
\end{equation}
Moreover,%
\begin{align}
\hat{V}^{\top}M_{X}\hat{V}  &  =Y_{1}^{\top}M_{Z}Y_{1}-(Y_{1}^{\top}M_{Z}%
Y_{1},0_{d_{y_{1}}\times d_{z_{1}}})(X^{\top}X)^{-1}(Y_{1}^{\top}M_{Z}%
Y_{1},0_{d_{y_{1}}\times d_{z_{1}}})^{\prime}\nonumber \\
&  =Y_{1}^{\top}M_{Z}Y_{1}-Y_{1}^{\top}M_{Z}Y_{1}(Y_{1}^{\top}M_{Z_{1}}%
Y_{1})^{-1}Y_{1}^{\top}M_{Z}Y_{1}\nonumber \\
&  =Y_{1}^{\top}(P_{Z}-P_{Z_{1}})Y_{1}(Y_{1}^{\top}M_{Z_{1}}Y_{1})^{-1}%
Y_{1}^{\top}M_{Z}Y_{1}, \label{P_L1_3}%
\end{align}
and%
\begin{align}
\hat{V}^{\top}M_{X}Y_{2}  &  =Y_{1}^{\top}M_{Z}Y_{2}-(Y_{1}^{\top}M_{Z}%
Y_{1},0_{d_{y_{1}}\times d_{z_{1}}})(X^{\top}X)^{-1}X^{\top}Y_{2}\nonumber \\
&  =Y_{1}^{\top}M_{Z}Y_{2}-Y_{1}^{\top}M_{Z}Y_{1}(Y_{1}^{\top}M_{Z_{1}}%
Y_{1})^{-1}Y_{1}^{\top}M_{Z_{1}}Y_{2}\nonumber \\
&  =Y_{1}^{\top}(P_{Z}-P_{Z_{1}})Y_{1}(Y_{1}^{\top}M_{Z_{1}}Y_{1})^{-1}%
Y_{1}^{\top}M_{Z_{1}}Y_{2}-Y_{1}^{\top}(P_{Z}-P_{Z_{1}})Y_{2},\nonumber \\
&  =Y_{1}^{\top}(P_{Z}-P_{Z_{1}})Y_{1}(\hat{\beta}_{ols}-\hat{\beta}_{2sls})
\label{P_L1_3b}%
\end{align}
where the second equality is by the partitioned regression formula on
$(X^{\top}X)^{-1}X^{\top}Y_{2}$, and the last equality is by\ (\ref{P_L1_2}%
).\ Combining the expressions in (\ref{P_L1_3}) and (\ref{P_L1_3b}) yields%
\begin{equation}
\hat{\rho}_{cf}=(\hat{V}^{\top}M_{X}\hat{V})^{-1}\hat{V}^{\top}M_{X}%
Y_{2}=(Y_{1}^{\top}M_{Z}Y_{1})^{-1}Y_{1}^{\top}M_{Z_{1}}Y_{1}(\hat{\beta
}_{ols}-\hat{\beta}_{2sls}). \label{P_L1_4}%
\end{equation}
The claim in (\ref{lm1_1}) follows from (\ref{P_L1_1}) and (\ref{P_L1_4}).

By the definition of $\widehat{Asv}(\hat{\rho}_{cf})$ in (\ref{CV_sigma}) and
the expression in (\ref{P_L1_3}),%
\begin{equation}
\widehat{Asv}(\hat{\rho}_{cf})=\hat{\sigma}_{u}^{2}(Y_{1}^{\top}M_{Z}%
Y_{1})^{-1}Y_{1}^{\top}M_{Z_{1}}Y_{1}(Y_{1}^{\top}(P_{Z}-P_{Z_{1}})Y_{1}%
)^{-1}. \label{P_L2_1}%
\end{equation}
Using (\ref{lm1_1}) and (\ref{P_L2_1}), we have%
\begin{align}
\hat{\rho}_{cf}^{\top}(\widehat{Asv}(\hat{\rho}_{cf}))^{-1}\hat{\rho}_{cf}  &
=\hat{\rho}_{cf}^{\top}\frac{Y_{1}^{\top}(P_{Z}-P_{Z_{1}})Y_{1}(Y_{1}^{\top
}M_{Z_{1}}Y_{1})^{-1}Y_{1}^{\top}M_{Z}Y_{1}}{\hat{\sigma}_{u}^{2}}\hat{\rho
}_{cf}\nonumber \\
&  =(\hat{\beta}_{ols}-\hat{\beta}_{2sls})^{\top}\frac{Y_{1}^{\top}M_{Z_{1}%
}Y_{1}(Y_{1}^{\top}M_{Z}Y_{1})^{-1}Y_{1}^{\top}(P_{Z}-P_{Z_{1}})Y_{1}}%
{\hat{\sigma}_{u}^{2}}(\hat{\beta}_{ols}-\hat{\beta}_{2sls}). \label{P_L2_2}%
\end{align}
Since $M_{Z}=M_{Z_{1}}-P_{Z}+P_{Z_{1}}$ and $\hat{Y}_{1}=P_{Z}Y_{1}$, we
derive%
\begin{align*}
&  (Y_{1}^{\top}(P_{Z}-P_{Z_{1}})Y_{1})^{-1}Y_{1}^{\top}M_{Z}Y_{1}(Y_{1}%
^{\top}M_{Z_{1}}Y_{1})^{-1}\\
&  =(Y_{1}^{\top}(P_{Z}-P_{Z_{1}})Y_{1})^{-1}Y_{1}^{\top}(M_{Z_{1}}%
-P_{Z}+P_{Z_{1}})Y_{1}(Y_{1}^{\top}M_{Z_{1}}Y_{1})^{-1}\\
&  =(Y_{1}^{\top}(P_{Z}-P_{Z_{1}})Y_{1})^{-1}-(Y_{1}^{\top}M_{Z_{1}}%
Y_{1})^{-1}\\
&  =(\hat{Y}_{1}^{\top}M_{Z_{1}}\hat{Y}_{1})^{-1}-(Y_{1}^{\top}M_{Z_{1}}%
Y_{1})^{-1},
\end{align*}
which together with (\ref{P_L2_2}) proves the second claim of the lemma.
\end{proof}

\bigskip

\begin{proof}
[Proof of Lemma \ref{lm2}]By the definition of $\hat{\theta}_{ols}$ and
$\hat{\theta}_{2sls}$, we can write
\begin{align}
\hat{\theta}_{ols}-\hat{\theta}_{2sls}=  &  (X^{\top}X)^{-1}X^{\top}M_{Z}%
Y_{2}+(X^{\top}X)^{-1}X^{\top}P_{Z}Y_{2}-\hat{\theta}_{2sls}\nonumber \\
=  &  (X^{\top}X)^{-1}X^{\top}M_{Z}Y_{2}+(X^{\top}X)^{-1}(X^{\top}%
P_{Z}X-X^{\top}X)\hat{\theta}_{2sls}\nonumber \\
=  &  (X^{\top}X)^{-1}(X^{\top}M_{Z_{1}}Y_{2}-X^{\top}M_{Z_{1}}X\hat{\theta
}_{2sls})\nonumber \\
&  -(X^{\top}X)^{-1}(X^{\top}(P_{Z}-P_{Z_{1}})Y_{2}-X^{\top}(P_{Z}-P_{Z_{1}%
})X\hat{\theta}_{2sls})\nonumber \\
=  &  (X^{\top}X)^{-1}(X^{\top}M_{Z_{1}}Y_{2}-X^{\top}M_{Z_{1}}X\hat{\theta
}_{2sls}). \label{P_L3_1}%
\end{align}
Applying the formula for the inverse of block matrix to $X^{\top}%
X=(Y_{1},Z_{1})^{\top}(Y_{1},Z_{1})$ and elementary matrix operations yields%
\begin{align*}
&  (X^{\top}X)^{-1}(X^{\top}M_{Z_{1}}Y_{2}-X^{\top}M_{Z_{1}}X\hat{\theta
}_{2sls})\\
&  =\left(
\begin{array}
[c]{c}%
(Y_{1}^{T}M_{Z_{1}}Y_{1})^{-1}\\
-(Z_{1}^{\top}Z_{1})^{-1}Z_{1}^{\top}Y_{1}(Y_{1}^{\top}M_{Z_{1}}Y_{1})^{-1}%
\end{array}
\right)  (Y_{1}^{\top}M_{Z_{1}}Y_{1}-Y_{1}^{\top}M_{Z_{1}}Y_{1}\hat{\beta
}_{2sls}).
\end{align*}
Note that
\[
Y_{1}^{\top}M_{Z_{1}}Y_{1}-Y_{1}^{\top}M_{Z_{1}}Y_{1}\hat{\beta}_{2sls}%
=Y_{1}^{\top}M_{Z_{1}}Y_{1}(\hat{\beta}_{ols}-\hat{\beta}_{2sls}).
\]
Therefore, we can further derive%
\begin{equation}
\hat{\theta}_{ols}-\hat{\theta}_{2sls}=\left(
\begin{array}
[c]{c}%
I_{d_{y_{1}}}\\
-(Z_{1}^{\top}Z_{1})^{-1}Z_{1}^{\top}Y_{1}%
\end{array}
\right)  (\hat{\beta}_{ols}-\hat{\beta}_{2sls}). \label{P_L3_2}%
\end{equation}

To show the first equality in (\ref{lm2_1}), we write
\begin{equation}
\hat{u}=Y_{2}-X\hat{\theta}_{cf}-\hat{V}\hat{\rho}_{cf}=Y_{2}-X\hat{\theta
}_{ols}+X(\hat{\theta}_{ols}-\hat{\theta}_{2sls})-\hat{V}\hat{\rho}_{cf}.
\label{P_L3_3}%
\end{equation}
Note that $X^{T}(Y_{2}-X\hat{\theta}_{ols})=0$ by the definition of
$\hat{\theta}_{ols}$. Some elementary matrix operations lead to%
\begin{align*}
\hat{V}^{T}(Y_{2}-X\hat{\theta}_{ols})  &  =Y_{1}^{\top}M_{Z}Y_{2}-Y_{1}%
^{\top}M_{Z}X\hat{\theta}_{ols}\\
&  =Y_{1}^{\top}M_{Z}Y_{2}-Y_{1}^{\top}M_{Z}Y_{1}\hat{\beta}_{ols}\\
&  =Y_{1}^{\top}M_{Z_{1}}Y_{2}-Y_{1}^{\top}M_{Z_{1}}Y_{1}\hat{\beta}%
_{ols}+Y_{1}^{\top}(P_{Z}-P_{Z_{1}})Y_{1}\hat{\beta}_{ols}-Y_{1}^{\top}%
(P_{Z}-P_{Z_{1}})Y_{2}\\
&  =Y_{1}^{\top}(P_{Z}-P_{Z_{1}})Y_{1}(\hat{\beta}_{ols}-\hat{\beta}_{2sls}).
\end{align*}
Hence, by Lemma \ref{lm1}
\begin{equation}
\hat{\rho}_{cf}^{T}\hat{V}^{T}(Y_{2}-X\hat{\theta}_{ols})=(\hat{\beta}%
_{ols}-\hat{\beta}_{2sls})^{\top}Y_{1}^{\top}M_{Z_{1}}Y_{1}(Y_{1}^{\top}%
M_{Z}Y_{1})^{-1}Y_{1}^{\top}(P_{Z}-P_{Z_{1}})Y_{1}(\hat{\beta}_{ols}%
-\hat{\beta}_{2sls}). \label{P_L3_4}%
\end{equation}
Moreover, observe that by (\ref{P_L3_2}) and Lemma \ref{lm1}
\begin{align*}
X(\hat{\theta}_{ols}-\hat{\theta}_{2sls})-\hat{V}\hat{\rho}_{cf}  &
=M_{Z_{1}}Y_{1}(\hat{\beta}_{ols}-\hat{\beta}_{2sls})-M_{Z}Y_{1}(Y_{1}^{\top
}M_{Z}Y_{1})^{-1}(Y_{1}^{\top}M_{Z_{1}}Y_{1})(\hat{\beta}_{ols}-\hat{\beta
}_{2sls})\\
&  =(M_{Z_{1}}Y_{1}-M_{Z}Y_{1}(Y_{1}^{\top}M_{Z}Y_{1})^{-1}Y_{1}^{\top
}M_{Z_{1}}Y_{1})(\hat{\beta}_{ols}-\hat{\beta}_{2sls})
\end{align*}
and by elementary matrix operations
\begin{align*}
&  (M_{Z_{1}}Y_{1}-M_{Z}Y_{1}(Y_{1}^{\top}M_{Z}Y_{1})^{-1}Y_{1}^{\top}%
M_{Z_{1}}Y_{1})^{\top}(M_{Z_{1}}Y_{1}-M_{Z}Y_{1}(Y_{1}^{\top}M_{Z}Y_{1}%
)^{-1}Y_{1}^{\top}M_{Z_{1}}Y_{1})\\
&  =Y_{1}^{\top}M_{Z_{1}}Y_{1}(Y_{1}^{\top}M_{Z}Y_{1})^{-1}Y_{1}^{\top
}M_{Z_{1}}Y_{1}-Y_{1}^{\top}M_{Z_{1}}Y_{1}\\
&  =Y_{1}^{\top}M_{Z_{1}}Y_{1}(Y_{1}^{\top}M_{Z}Y_{1})^{-1}Y_{1}^{\top}%
(P_{Z}-P_{Z_{1}})Y_{1}.
\end{align*}
Therefore, we have%
\begin{align}
&  (X(\hat{\theta}_{ols}-\hat{\theta}_{2sls})-\hat{V}\hat{\rho}_{cf}%
)^{T}(X(\hat{\theta}_{ols}-\hat{\theta}_{2sls})-\hat{V}\hat{\rho}%
_{cf})\nonumber \\
&  =(\hat{\beta}_{ols}-\hat{\beta}_{2sls})^{\top}Y_{1}^{\top}M_{Z_{1}}%
Y_{1}(Y_{1}^{\top}M_{Z}Y_{1})^{-1}Y_{1}^{\top}(P_{Z}-P_{Z_{1}})Y_{1}%
(\hat{\beta}_{ols}-\hat{\beta}_{2sls}). \label{P_L3_5}%
\end{align}
Collecting the results in (\ref{P_L3_3}), (\ref{P_L3_4}), and (\ref{P_L3_5}),
we obtain
\begin{align*}
n^{-1}\hat{u}^{\top}\hat{u}=  &  n^{-1}(Y_{2}-X\hat{\theta}_{ols})^{T}%
(Y_{2}-X\hat{\theta}_{ols})-2n^{-1}\hat{\rho}_{cf}^{T}\hat{V}^{T}(Y_{2}%
-X\hat{\theta}_{ols})\\
&  +n^{-1}(X(\hat{\theta}_{ols}-\hat{\theta}_{2sls})-\hat{V}\hat{\rho}%
_{cf})^{T}(X(\hat{\theta}_{ols}-\hat{\theta}_{2sls})-\hat{V}\hat{\rho}_{cf})\\
=  &  \hat{\sigma}_{ols}^{2}-n^{-1}(\hat{\beta}_{ols}-\hat{\beta}%
_{2sls})^{\top}Y_{1}^{\top}M_{Z_{1}}Y_{1}(Y_{1}^{\top}M_{Z}Y_{1})^{-1}%
Y_{1}^{\top}(P_{Z}-P_{Z_{1}})Y_{1}(\hat{\beta}_{ols}-\hat{\beta}_{2sls})\\
=  &  \hat{\sigma}_{ols}^{2}-n^{-1}(\hat{\beta}_{ols}-\hat{\beta}%
_{2sls})^{\top}((Y_{1}^{\top}(P_{Z}-P_{Z_{1}})Y_{1})^{-1}-(Y_{1}^{\top
}M_{Z_{1}}Y_{1})^{-1})^{-1}(\hat{\beta}_{ols}-\hat{\beta}_{2sls})\\
=  &  \hat{\sigma}_{ols}^{2}\left(  1-\frac{t_{H,n}\left(  \hat{\sigma}%
_{ols}^{2},\hat{\sigma}_{ols}^{2}\right)  }{n}\right)  ,
\end{align*}
which together with the definition of $t_{H_{1}}$ proves the first equality in
(\ref{lm2_1}).

To show the second equality in (\ref{lm2_1}), note that by Lemma \ref{lm1}
\begin{equation}
\hat{u}=Y_{2}-X\hat{\theta}_{cf}-\hat{V}\hat{\rho}_{cf}=Y_{2}-X\hat{\theta
}_{2sls}-\hat{V}\hat{\rho}_{cf}. \label{P_L3_6}%
\end{equation}
By the definitions of $\hat{\beta}_{ols}$ and $\hat{\beta}_{2sls}$, we have%
\begin{align}
\hat{V}^{T}(Y_{2}-X\hat{\theta}_{2sls})=  &  Y_{1}^{\top}M_{Z}Y_{2}%
-Y_{1}^{\top}M_{Z}X\hat{\theta}_{2sls}\nonumber \\
=  &  Y_{1}^{\top}M_{Z}Y_{2}-Y_{1}^{\top}M_{Z}Y_{1}\hat{\beta}_{2sls}%
\nonumber \\
=  &  Y_{1}^{\top}M_{Z_{1}}Y_{2}-Y_{1}^{\top}M_{Z_{1}}Y_{1}\hat{\beta}%
_{2sls}-Y_{1}^{\top}(P_{Z}-P_{Z_{1}})Y_{2}+Y_{1}^{\top}(P_{Z}-P_{Z_{1}}%
)Y_{1}\hat{\beta}_{2sls}\nonumber \\
=  &  Y_{1}^{\top}M_{Z_{1}}Y_{1}(\hat{\beta}_{ols}-\hat{\beta}_{2sls}).
\label{P_L3_7}%
\end{align}
Therefore,%
\begin{equation}
\hat{\rho}_{cf}^{T}\hat{V}^{T}(Y_{2}-X\hat{\theta}_{2sls})=(\hat{\beta}%
_{ols}-\hat{\beta}_{2sls})^{\top}Y_{1}^{\top}M_{Z_{1}}Y_{1}(Y_{1}^{\top}%
M_{Z}Y_{1})^{-1}Y_{1}^{\top}M_{Z_{1}}Y_{1}(\hat{\beta}_{ols}-\hat{\beta
}_{2sls}). \label{P_L3_8}%
\end{equation}
Moreover,%
\[
\hat{\rho}_{cf}^{\top}\hat{V}^{T}\hat{V}\hat{\rho}_{cf}=(\hat{\beta}%
_{ols}-\hat{\beta}_{2sls})^{\top}Y_{1}^{\top}M_{Z_{1}}Y_{1}(Y_{1}^{\top}%
M_{Z}Y_{1})^{-1}Y_{1}^{\top}M_{Z_{1}}Y_{1}(\hat{\beta}_{ols}-\hat{\beta
}_{2sls}),
\]
which combined with (\ref{P_L3_6}), (\ref{P_L3_7})\ and (\ref{P_L3_8}) implies
that%
\begin{align*}
n^{-1}\hat{u}^{\top}\hat{u}=  &  n^{-1}(Y_{2}-X\hat{\theta}_{2sls})^{\top
}(Y_{2}-X\hat{\theta}_{2sls})+n^{-1}\hat{\rho}_{cf}^{\top}\hat{V}^{T}\hat
{V}\hat{\rho}_{cf}-2n^{-1}\hat{\rho}_{cf}^{T}\hat{V}^{T}(Y_{2}-X\hat{\theta
}_{2sls})\\
=  &  \hat{\sigma}_{2sls}^{2}-(\hat{\beta}_{ols}-\hat{\beta}_{2sls})^{\top
}Y_{1}^{\top}M_{Z_{1}}Y_{1}(Y_{1}^{\top}M_{Z}Y_{1})^{-1}Y_{1}^{\top}M_{Z_{1}%
}Y_{1}(\hat{\beta}_{ols}-\hat{\beta}_{2sls})\\
=  &  \hat{\sigma}_{2sls}^{2}-n^{-1}(\hat{\beta}_{ols}-\hat{\beta}%
_{2sls})^{\top}Y_{1}^{\top}M_{Z_{1}}Y_{1}(Y_{1}^{\top}M_{Z}Y_{1})^{-1}%
Y_{1}^{\top}(P_{Z}-P_{Z_{1}})Y_{1}(\hat{\beta}_{ols}-\hat{\beta}_{2sls})\\
&  -n^{-1}(\hat{\beta}_{ols}-\hat{\beta}_{2sls})^{\top}Y_{1}^{\top}M_{Z_{1}%
}Y_{1}(Y_{1}^{\top}M_{Z}Y_{1})^{-1}Y_{1}^{\top}M_{Z}Y_{1}(\hat{\beta}%
_{ols}-\hat{\beta}_{2sls})\\
=  &  \hat{\sigma}_{2sls}^{2}-\hat{\sigma}_{2sls}^{2}\frac{t_{H,n}\left(
\hat{\sigma}_{2sls}^{2},\hat{\sigma}_{2sls}^{2}\right)  }{n}-(\hat{\beta
}_{ols}-\hat{\beta}_{2sls})^{\top}\frac{Y_{1}^{\top}M_{Z_{1}}Y_{1}}{n}%
(\hat{\beta}_{ols}-\hat{\beta}_{2sls}).
\end{align*}
This together with the definitions of $t_{H_{2}}$ and $H_{n}$ completes the proof.
\end{proof}

\bigskip

\begin{proof}
[Proof of Lemma \ref{lm3}]Note that by $\hat{\beta}_{ols}\neq \hat{\beta
}_{2sls}$ and (\ref{PD_1}), we have
\begin{equation}
\hat{\sigma}_{ols}^{2}t_{H_{1}}=\hat{\sigma}_{2sls}^{2}t_{H_{2}}=\hat{\sigma
}_{u}^{2}t_{CF}=(\hat{\beta}_{ols}-\hat{\beta}_{2sls})^{\top}((\hat{Y}%
_{1}^{\top}M_{Z_{1}}\hat{Y}_{1})^{-1}-(Y_{1}^{\top}M_{Z_{1}}Y_{1})^{-1}%
)^{-1}(\hat{\beta}_{ols}-\hat{\beta}_{2sls})>0 \label{P_L4_1}%
\end{equation}
and
\begin{equation}
\hat{\sigma}_{2sls}^{2}H_{n}=n^{-1}(\hat{\beta}_{ols}-\hat{\beta}%
_{2sls})^{\top}(Y_{1}^{\top}M_{Z_{1}}Y_{1})(\hat{\beta}_{ols}-\hat{\beta
}_{2sls})>0. \label{P_L4_2}%
\end{equation}
Therefore, the equalities in (\ref{lm2_1}) of Lemma \ref{lm2} imply that
\begin{equation}
\hat{\sigma}_{2sls}^{2}>\hat{\sigma}_{ols}^{2}>\hat{\sigma}_{u}^{2}.
\label{P_L4_3}%
\end{equation}
Combining (\ref{P_L4_1}) and (\ref{P_L4_3}) yields $t_{CF}>t_{H_{1}}>t_{H_{2}%
}$.\ By (\ref{P_L4_3}) and $(Y_{1}^{\top}M_{Z_{1}}Y_{1})^{-1}>0$, we can show
that
\[
\hat{\sigma}_{2sls}^{2}(\hat{Y}_{1}^{\top}M_{Z_{1}}\hat{Y}_{1})^{-1}%
-\hat{\sigma}_{ols}^{2}(Y_{1}^{\top}M_{Z_{1}}Y_{1})^{-1}>\hat{\sigma}%
_{2sls}^{2}[(\hat{Y}_{1}^{\top}M_{Z_{1}}\hat{Y}_{1})^{-1}-(Y_{1}^{\top
}M_{Z_{1}}Y_{1})^{-1}],
\]
which together with $\hat{\beta}_{ols}\neq \hat{\beta}_{2sls}$ and the
definitions of $t_{H_{2}}$ and $t_{H_{3}}$ implies that $t_{H_{2}}>t_{H_{3}}$.
\end{proof}

\end{document}